\documentclass[sigconf]{acmart}

\usepackage{xspace}
\usepackage{balance}

\usepackage{tabulary}
\usepackage{tabularx}
\usepackage{epsfig,endnotes}
\usepackage{booktabs}
\usepackage[mathcal]{euscript}

\usepackage{paralist}

\usepackage{graphicx,citesort}
\usepackage{epsfig,epsf,url,amssymb}
\usepackage{tabularx}
\usepackage{slashbox}
\newtheorem{defi}{Definition}
\newtheorem{lem}{Lemma}
\newtheorem{thrm}{Theorem}

\usepackage{libertine}
\usepackage[noend]{algpseudocode}
\makeatletter
\def\BState{\State\hskip-\ALG@thistlm}
\makeatother

\usepackage{algorithmicx}
\usepackage{algpseudocode}
\usepackage{amsmath}

\usepackage{float}

\usepackage{listings}

\usepackage{pifont}

\usepackage{bbding}
\usepackage{wasysym}
\usepackage{amssymb}

\usepackage{xcolor}
\definecolor{lightgray}{rgb}{.9,.9,.9}
\definecolor{darkgray}{rgb}{.4,.4,.4}
\definecolor{purple}{rgb}{0.65, 0.12, 0.82}

\lstdefinelanguage{JavaScript}{
  keywords={typeof, new, true, false, catch, function, return, null, catch, switch, var, if, in, while, do, else, case, break},
  keywordstyle=\color{blue}\bfseries,
  ndkeywords={class, export, boolean, throw, implements, import, this, extends},
  ndkeywordstyle=\color{darkgray}\bfseries,
  identifierstyle=\color{black},
  sensitive=false,
  comment=[l]{//},
  morecomment=[s]{/*}{*/},
  commentstyle=\color{purple}\ttfamily,
  stringstyle=\color{red}\ttfamily,
  morestring=[b]',
  morestring=[b]",
  moredelim=[is][\underbar]{_}{_}
}

\def\sys{\textsc{DeterFox}\xspace}

\usepackage[disable]{todonotes}

\newenvironment{icompact}{
  \begin{list}{$\bullet$}{
    \itemindent -.05em
    \parsep 0pt plus 1pt
    \partopsep 0pt plus 1pt
    \topsep 2pt plus 2pt minus 1pt
    \itemsep 0pt plus 1.3pt
    \parskip 0pt plus 2pt
    \leftmargin 0.13in}
       }
  {\normalsize\end{list}}

\usepackage{todonotes}

\begin{document}
\title{{Deterministic Browser}}


\setcopyright{none}
\begin{abstract}
Timing attacks 
 have been a continuous threat to users' privacy in modern browsers. To mitigate such attacks, 
  existing approaches, such as Tor Browser and Fermata, add jitters to the browser clock so that an attacker cannot accurately measure an event.  However, such defenses only raise the bar for an attacker but do not fundamentally mitigate timing attacks, i.e., it just takes longer than previous to launch a timing attack.  

In this paper, we propose a novel approach, called deterministic browser, which can provably prevent timing attacks in modern browsers.  Borrowing from Physics, we introduce several concepts, such as an observer and a reference frame.  Specifically, a snippet of JavaScript, i.e., an observer in JavaScript reference frame, will always obtain the same, fixed timing information so that timing attacks are prevented; at contrast, a user, i.e., an oracle observer, will perceive the JavaScript differently and do not experience the performance slowdown.  We have implemented a prototype called \sys and our evaluation shows that the prototype can defend against browser-related timing attacks.



\end{abstract}

\author{Yinzhi Cao}
\affiliation{%
  \institution{Lehigh University}
  \city{Bethlehem} 
  \state{PA} 
}
\email{yinzhi.cao@lehigh.edu}

\author{Zhanhao Chen}
\affiliation{%
  \institution{Lehigh University}
  \city{Bethlehem} 
  \state{PA} 
}
\email{zhc416@lehigh.edu}

\author{Song Li}
\affiliation{%
  \institution{Lehigh University}
  \city{Bethlehem} 
  \state{PA} 
}
\email{sol315@lehigh.edu}

\author{Shujiang Wu}
\affiliation{%
  \institution{Lehigh University}
  \city{Bethlehem} 
  \state{PA} 
}
\email{shw316@lehigh.edu}

\keywords{Determinism, Web Browser, Timing Side-channel Attack}

\ccsdesc[500]{Security and privacy~Browser security}
\ccsdesc[300]{Security and privacy~Web application security}

\maketitle

\section{Introduction}

Timing attacks have continuously posed a threat to modern web browsers, violating users' privacy.  For example, an adversary can infer the size of an external, cross-site resource based on the loading time~\cite{VanGoethem:2015:CST:2810103.2813632,vanrequest};  a website can fingerprint the type of the browser based on the performance of JavaScript~\cite{mowery2011fingerprinting,mulazzani2013fast}; 
 two adversaries can talk to each other 
  via a covert channel~\cite{Oren:2015:SSP:2810103.2813708,VanGoethem:2015:CST:2810103.2813632}.

Faced with browser-related timing attacks, researchers from both industry and academia have proposed solutions to mitigate such attacks by adding noises or called jitters to the time available to the attackers and decreasing the time precision.  For example, Tor Browser~\cite{torbrowser}, an industry pioneer in protecting users' privacy, reduces the resolution of $performance.now$, a fine-grained JavaScript clock, to 100ms; following Tor Browser,  major browsers, such as Chrome and Firefox, also reduces the resolution to 5$\mu$s.  Similarly, Kohlbrenner and Shacham proposed  Fermata~\cite{fuzzyfox}, a solution that introduces fuzzy time~\cite{Hu:1992:RTC:2699806.2699810} into browser design and then degrades not only the explicit clock like $performance.now$ but also other implicit clocks like $setTimeout$ in the browser. 

However, the aforementioned prior work---which reduces the browser's clock resolution---only raises the bar for an attacker but do not fundamentally mitigate timing attacks.  Even if such work is deployed, it just takes longer time for an adversary to launch an aforementioned timing attack and obtain the information that she needs.  Say, an adversary tries to infer a resource's size based on the loading time, i.e., launching a side-channel attack.  All the existing defenses only limit the bandwidth of such  a side channel as acknowledged by Kohlbrenner and Shacham~\cite{fuzzyfox}. 
 Consider JavaScript performance fingerprinting as another example. Jitters that are added by existing defenses can be averaged out, when the attack is performed longer or in multiple runs.  

Apart from the prior work targeting browser-related timing attacks, another direction that prevents lower-level timing attacks, such as L2 cache attacks, is to introduce determinism.  Examples of such approaches are Deterland~\cite{deterland}, StopWatch~\cite{Li:2014:SCA:2689660.2670940,DBLP:conf/dsn/LiGR13}, and $\lambda^{PAR}_{SEC}$~\cite{DBLP:conf/csfw/ZdancewicM03}.  However, Deterland still adopts statistical, non-deterministic solutions for external events by grouping events together, which only limits the bandwidth of an external timing attack.  StopWatch cannot prevent internal timing attacks, and needs virtual machine replication for I/O events, which is not applicable at a higher level, such as a browser.    $\lambda^{PAR}_{SEC}$ is a new programming language with ensured security property, but requires that all the existing programs are rewritten and follow their language specification.  In sum, it still remains unclear how to apply determinism in real-world systems even at lower level let alone web browsers so as to prevent timing attacks.



In this paper, we propose deterministic browser, the {\it first} approach that introduces determinism into web browsers and provably prevents browser-related timing attacks.
 Both challenges and opportunities arise in deterministic browser.  One challenge is that JavaScript, the dominant web language, is event-driven, i.e., 
 JavaScript engine may be waiting for events without executing any statements; accordingly, one opportunity is that JavaScript is singled-threaded---specifically no event can interrupt the JavaScript execution.\footnote{The statement holds, even if we take WebWorkers, a new HTML5 standard, into consideration.}

To better explain deterministic browser, let us revisit timing attacks.  A timing attack is that an adversary tries to measure the duration of a target, secret event (called a ``target secret'' in the paper) using a reference clock of the physical time.  The target secret is an event that the adversary does not know how long it takes to finish, e.g., the parsing time of a cross-origin script.  The reference clock is used to measure the target secret, which could be explicit, e.g., $performance.now$, or implicit, e.g., $setTimeout$ and $requestAnimationFrame$.  An implicit reference clock is 
%
%
 based on an event that the adversary knows how long it takes to finish.  
  The duration of such known event is defined as an implicit clock tick, usually much smaller than the duration of the target secret. 
  
  A {\it necessary condition} of launching a successful timing attack is that all three key elements---i.e., ($i$) an adversary, ($ii$) a target secret, and ($iii$) a reference clock---have to co-exist at the same time in a runtime environment.  An adversary is the subject, a target secret is the object, and a reference clock is the tool that the subject used to steal the object.  Therefore,  a natural solution is to remove one or more of the three key elements from the runtime environment and break the necessary condition.  However, all three elements are essential in a browser and hard to remove.  An adversary just looks like a normal client-side JavaScript.  The detection of such JavaScript is a different research direction and based on past research on JavaScript malware detection~\cite{zozzle,JShield}, false positives and negatives are hard to avoid. 
   A target secret is a common browser operation such as script parsing---i.e., the removal of target secrets will break the browser functionality.  We can remove\footnote{Strictly speaking, because a clock is an essential concept, the so-called removal of a clock is just to define a new clock in which the ticking unit is zero.} an explicit reference clock of the physical time, such as $performance.now$, but many implicit clocks like $setTimeout$ and using JavaScript execution as minor clocks~\cite{fuzzyfox} also exist and contribute to essential browser functionalities.

In deterministic browser, instead of removing key elements from the entire runtime environment, i.e., the web browser, we break down the web browser into several smaller units and remove one or more different key elements from these smaller units.  That is, from a macro perspective, all three key elements still exist in the web browser; from a micro perspective, at most two elements exist in one smaller unit.   
 Specifically, the smaller unit that we introduce is called reference frames (RFs),\footnote{Note that a reference frame is not an HTML or JavaScript frame: An HTML or JavaScript frame may contain many reference frames.  To avoid confusion, we often use the abbreviation, i.e., RF, for reference frame in the paper. } a new, abstract concept borrowed from Physics.  Each RF has one independent clock and sometimes an observer---e.g., a JavaScript program possibly controlled by the adversary.  Just as in Physics, the job of the observer is to measure the duration of an event in the RF, i.e., making two observations at the start and end of the event, obtaining two timestamps and calculating the interval.  In web browser, one important yet different property of a RF is that one and only one event---e.g., a target secret or an implicit clock tick event---can be executed in one RF; different events are executed in different RFs separately.  (As discussed in Section~\ref{subsec:concepts}, RFs are usually implemented by OS level threads.)  We now look at different RF examples and explain how to remove key elements of a timing attack from these RFs. In general, RFs can be categorized into two types: with and without a target secret. 





 First, let us say a RF is executing a target secret---that is, all three key elements co-exist in this RF.  
 %
  JavaScript execution engine is such an example, in which the target secret is the execution performance of JavaScript. 
  In such case, we need to redefine the clock in the RF such that when an observer measures the duration of any event, it can always calculate the duration based on what it has already known without looking at the clock.  A RF with such a clock is defined as deterministic (Definition~\ref{def:det}) in our paper.  Here our definition of determinism---which provably prevents timing attacks shown in Theorem~\ref{theorem:detdefeattiming}---is broader than the existing definition, such as the one used in DeterLand~\cite{deterland}.  We further show that in the context of JavaScript RF, our definition can boil down to a simple form (Theorem~\ref{th:deteq}), and then be specialized to the existing definition.  Particularly, the clock in JavaScript RF ticks based on the number of executed opcodes. 
  


Second, let us say a RF has no target secret, e.g., a RF executing a tick event for an implicit reference clock.  Concrete examples are like a RF executing $setTimeout$, one dealing with an HTTP request, and another executing $requestAnimationFrame$. In such cases, only two elements exist in the RF---there is no target secret for the adversary to steal.  That is, the clock in such RFs can follow the physical time. 
%
 However, the clock information in such RFs can be communicated to a deterministic RF that has a target secret, hence being used for measurement.  Therefore, when a deterministic RF receives information from another (non-deterministic) RF, we need to make the recipient RF remain deterministic (Lemma~\ref{thereom:comm})---the technique of such communication between RFs is defined as  deterministic communication (Definition~\ref{def:comm}).  Particularly, as shown in Section~\ref{subsec:comm}, we introduce a priority queue to replace the original event queue, i.e., the central communication data structure, in the JavaScript event model.  Our new priority queue synchronizes clocks in two RFs from the viewpoint of the recipient observer.
  In other words, when the recipient observer 
   sees the communication message conveying the time from other RF, e.g., the physical time, the conveyed time is the same as its own clock.


Apart from the aforementioned categorization of RFs based on whether a target secret exists, RFs can also be classified based on the existence of an observer, one with an observer called a main RF and one without called  an auxiliary RF.  The clock in a RF is in an undefined state when an observer is absent, because only observers can obtain timing information, i.e., making observations.  Then, the clock will become 
 available when the observer is back, usually when a RF communication happens. For example, when a JavaScript RF is waiting for an event to finish, the clock in the JavaScript RF is undefined and determined by the event to finish. For another example, when a $setTimeout$ RF is waiting for some amount of time, the clock is also undefined and determined when the RF communicates with the JavaScript RF.

Another important property of deterministic browser---which is just similar to Physics---is that 
 different observers perceive different time elapse in their own RFs. 
 For example, the observer in JavaScript RF, i.e., a snippet of JavaScript code possibly controlled by an adversary,
 %
  will always obtain the same performance information when measuring itself so that timing attacks are prevented. At contrast, an oracle observer, e.g., the user of the browser, will perceive the JavaScript execution just as normal as in a legacy browser so that she does not experience the performance slowdown as what the JavaScript observer does.

We have implemented a prototype of our deterministic browser called \sys by modifying a legacy Firefox browser.  Our evaluation shows that \sys can defend against browser-related timing attacks in the literature 
 and is compatible with existing websites.  Note that our implementation is open-source---the repository and a video demo can be found at this website (\url{http://www.deterfox.com}).



\section{Threat Model} \label{sec:threatmodel}

We present our threat model 
 from three aspects: in-scope attacks, a motivating example, and out-of-scope attacks. 
\vspace{0.05in}

\noindent{\bf In-scope Attacks.} \hspace{0.01in} In this paper, we include 
 browser-related timing attacks in our threat model. As described, such attacks are made possible because an adversary can measure the duration of a target secret at client-side browsers using a reference clock.   The adversary is the client-side JavaScript and the reference clock could be the physical one defined by $performance.now$ or other implicit ones, such as $setTimeout$.  The target secret varies according to the specific attack scenario, and now 
 let us use  the following three examples to explain it. 
 
 


%
%


\begin{icompact}
\item {JavaScript Performance Fingerprinting.} \hspace{0.01in} JavaScript performance fingerprinting~\cite{mowery2011fingerprinting,mulazzani2013fast} is one special case of browser fingerprinting where an adversary executes a certain snippet of JavaScript code and fingerprints the browser based on how long the execution takes.
 In such attacks, the {\it target secret} is the performance of JavaScript execution on the specific machine. 
 

\item {Inference of Cross-origin Resource via Loading Time.} \hspace{0.01in}  
An adversary website may load resources from a third-party website via a script tag and measure the parsing time by the browser until an error event is triggered.  By doing so, the adversary can infer the size of the response and thus other private information, such as the number of Twitter friends~\cite{VanGoethem:2015:CST:2810103.2813632,vanrequest}.  This is also called a timing side channel attack.  In this example, the {\it target secret} is the loading of the cross-origin response, which the adversary does not know due to the same-origin policy. 

\item {Inference of Image Contents via SVG Filtering.} \hspace{0.01in} 
 Stone~\cite{stone2013pixel} shows that an adversary can apply an SVG filter on an image and infer the contents based on how long the filtering process takes.  By doing so, the adversary can steal the pixels of a cross-origin image or browsing history.  In this example, the {\it target secret} is the performance of the SVG filter.  Note that the {\it reference clock} used in this attack is $requestAnimationFrame$. 
 

\end{icompact}

\vspace{0.05in}

\noindent{\bf A Motivating Example.} \hspace{0.01in} Now let us look at a motivating example in Figure~\ref{fig:example}.
 The example contains two versions of attacks: synchronous (Line 1--21) and asynchronous (Line 22--28) ones.  The synchronous version uses $performance.now$ as the reference clock and the asynchronous one uses $setInterval$. We will discuss how determinism prevents these two versions in Section~\ref{subsec:defeatsyncattack} and~\ref{subsec:defeatasyncattack} respectively.

The synchronous attack (Line 1--21) adopts a so-called clock-edge technique invented by Kohlbrenner and Shacham~\cite{fuzzyfox}, which measures a minor clock, e.g., a simple, cheap JavaScript operation like $count++$, by repeating the operation between two edges of a major clock, e.g., the physical clock, for multiple times.  Then, one can calculate the minor clock by dividing the granularity of major clock by the number of executed operations. 
 Here is the detail of the attack.  The $nextEdge$ function (Line 1--8) tries to find the next edge of the major clock ($performance.now$) and count the number of executed operations towards the next major clock edge.  
 For a JavaScript fingerprinting attack (Line 9--12), one can first skip the rest of the current major clock period (Line 11), and then calculate the number of operations in a full major clock period (Line 12) for fingerprinting.  
  For a side-channel attack (Line 13--21), one can find the fingerprint (Line 14), start from a new edge (Line 15), and then measure a target secret using the major clock (Line 16--18).  
  The remaining cycles in the minor clock are counted (Line 19), and the duration for the target secret is calculated by combining the minor and major clocks (Line 20).  Note that $grain$ is the granularity of the major clock.  



The asynchronous attack (Line 22--28) adopts an implicit clock, e.g., $setTimeout$, $setInterval$ or $requestAnimationFrame$, to measure an asynchronous target secret, e.g., the time to parse a third-party response. First, the $sideChannelAsync$ function (Line   
25--28) invokes an asynchronous target secret function---for simplicity, we hide the details in the $targetSecretAsync$ function with only a $callback$ function (Line 26).  Then, the implicit clock is invoked via a $setInterval$ function with $countFunc$ as a callback (Line 27). The $countFunc$ will be invoked periodically with a $u$ interval until the $callback$ is invoked, i.e., the target secret finished execution.  Then, one can calculate the duration of the target secret based on how many times $countFunc$ is executed (Line 24). 


\begin{figure}
\scriptsize
    \begin{lstlisting}[language={JavaScript},xleftmargin=2em,frame=none,framexleftmargin=2em,numbersep=7pt,numbers=left,keepspaces=true]
function nextEdge() {
  start = performance.now();
  count = 0;
  do {
  	count++;
  } while (start==performance.now());
  return count;
}
function fingerprinting() { 
  nextEdge();
  fingerprint = nextEdge();
}
function sideChannelSync() {  
  fingerprint = fingerprinting();
  nextEdge();
  start = performance.now();
  targetSecret();  
  stop = performance.now();
  remain = nextEdge();
  duration = (stop - start) + (fingerprint-remain)/fingerprint*grain;
}
total = 0;
function countFunc(){total++;}
function callback(){duration = total*u;}
function sideChannelAsync() {
 targetSecretAsync(callback);
 setInterval(countFunc, u);
}
\end{lstlisting}
\vspace{-0.2in}
\caption{A Motivating Example.
 } \label{fig:example}
 \vspace{-0.15in}
\end{figure}

\vspace{0.05in}

\noindent{\bf Out-of-scope Attacks.} \hspace{0.01in} We now look at out-of-scope attacks. 
 Particularly, we restrict the target secret in the timing attacks to be browser-related and exclude these depending on other parties, such as the web server and the user.  Such restriction is natural because we only make the browser deterministic but not others. 

\begin{icompact}
\item Server-side timing attacks. \hspace{0.01in} In a server-side timing attack, a malicious client infers a secret---e.g., some information that the client cannot access without login---at the server side~\cite{side-channel-leaks-in-web-applications-a-reality-today-a-challenge-tomorrow}.  Examples are that a user can infer the number of private photos in a hidden gallery based on the server's response time.  To prevent such attacks, one can rely on existing defenses~\cite{DBLP:conf/ndss/BackesDK13} by normalizing web traffic or changing the server-side code. 
\item User-related timing attacks.  \hspace{0.01in} User-related timing attacks refer to, for example, biometric fingerprinting, which measures the user's behavior, such as keystroke and mouse move, and identifies users based on these biometrics. 
%
 Many of such fingerprinting tactics, such as keystroke dynamics, are also related to time.  Such attacks are out of scope, because the target secret, e.g., mouse move, comes from the user.  
 
 
 

\end{icompact}



\section{Deterministic Browser}  \label{sec:def}
In this section, we introduce several general definitions. 


\subsection{Concepts} \label{subsec:concepts}
Borrowing from Physics, we first introduce the concepts used in our deterministic browser.  

A {\em reference frame} (RF), in the context of our deterministic browsers (referred as our context), is an abstract concept with a clock that ticks based on certain criteria, e.g., the real-world physical time, the execution of JavaScript, and the parsing of HTML.  Concretely, we can consider that a RF can be implemented via an OS-level thread in web browsers, although there are no strict one-one mappings between threads and RFs in web browsers.  Specifically, some browser kernel threads may not be RFs and two threads, if one runs immediately after another, may belong to the same RF. 

Let us take a browser API, e.g., $setTimeout$ or $setInterval$, and its implementation in Firefox, as a concrete example of creating RFs.  The JavaScript engine of Firefox, a RF in our context, is running in the main thread of Firefox.  When JavaScript code calls $setTimeout$ or $setInterval$ under the $window$ object,  a new timer thread---i.e., a new RF---is created.  After the specified time in $setTimeout$ or $setInterval$ passes, the timer thread will communicate with the main thread with an event.  That said, the concept of RF is natural and modern browsers have already provided implementations.  Note that although we use Firefox as an example, the concept of RF is abstract and general.  In a multiple-process browser, like Google Chrome, RFs could be represented by processes or threads in multiple processes.  

The concept of RF defines a unit that is smaller than all other web containers, such as HTML frames, origins, JavaScript runtime, and browser tabs.  The advantage of a RF is that we can easily separate three key elements---an adversary, a target secret, and a reference clock---of a timing attack. For example, because a RF is singled-threaded, i.e., can execute only one event, a target secret and an implicit clock tick are separated in different RFs naturally. 




For each RF, one important concept is called {\em an observer}.  In our context, {\em an observer} is defined as a Turing complete program that can make observations, e.g., access the data belonging to the RF or measure the execution status of itself.  Some RFs have observers, such as a JavaScript RF with JavaScript program, and some, such as DOM RFs, do not. A RF with an observer is called a {\it main RF}, and a RF without an observer an {\it auxiliary RF}.  
 For the purpose of defining determinism, we can consider virtual observers residing in auxiliary RFs. 





\subsection{Definition of Determinism}

Let us introduce the definition of determinism, which includes two parts: (i) how to define a deterministic RF, and (ii) how to define a deterministic communication between RFs.   Next, we will show how determinism can prevent timing attacks. 
%
%
 Now we first define a deterministic RF in Definition~\ref{def:det}.


\begin{defi}
\label{def:det} 
{\bf{(Deterministic RF)}} Given a reference frame ($RF$) and an observer ($Ob$)---no matter active or virtual, we define a $RF$ as deterministic if and only if the following holds:

When the $Ob$ makes two observations $O_1$ and $O_2$ at timestamps $t_1$ and $t_2$ measured by the internal clock, $t_2-t_1$ can be represented as a function of $O_1$ and $O_2$, i.e., $t_2-t_1 = f(O_1, O_2)$.
\end{defi}


The core of Definition~\ref{def:det} is that when an observer makes observations, e.g., access the data of the RF or measure the execution of JavaScript, the observer cannot obtain additionally timing information other than what it has already known.  In other words, based on the observations in the RF, the observer can directly deduce the timing, e.g., $t_2-t_1$, without looking at the clock in the RF.  

For the purpose of explanation, let us look at one toy example of defining deterministic RFs.  Specifically, the RF clock ticks based on the following rule: Every time an integer variable $x$ is incremented by 1, the clock ticks by an atomic unit.  This RF is deterministic because the clock is directly defined as a function of a variable, i.e., an observation that the observer can make.  As we can expect, this RF can prevent timing attacks (assuming that external clocks are handled by Definition~\ref{def:comm} below). However, because the clock may not tick if the observer does not increment it, browser functionalities will be broken. 





Beside the internal clock, an observer can also 
%
 obtain other clocks via RF communication.  
  So we define a deterministic communication in Definition~\ref{def:comm}.




\begin{defi} 
\label{def:comm}
{\bf{(Deterministic Communication)}} Given two RFs ($RF^{receiver}$, a deterministic RF, and $RF^{sender}$, another RF) and an observer ($Ob$) in $RF^{receiver}$, we define the communication from $RF^{sender}$ to $RF^{receiver}$ as deterministic if and only if either of the following holds:

(1) When the $Ob$ from $RF^{receiver}$ makes an observation in $RF^{sender}$ at the timestamp $t^{{receiver}}$ in $RF^{receiver}$ (i.e., a communication from $RF^{sender}$ to $RF^{receiver}$ happens), at that moment, the timestamp $t^{sender}$ in $RF^{sender}$ equals to $t^{{receiver}}$ ($t^{{sender}}=t^{{receiver}}$);

(2) $RF^{sender}$ is deterministic.

\end{defi}

Definition~\ref{def:comm} gives two conditions for a deterministic communication between a sender and a receiver, i.e., either (1) two clocks are synchronized or (2) both RFs are deterministic. 
%
%
%
  Note that in order to synchronize two clocks (i.e, the first condition), it is required that the clock in the receiver is behind the one in the sender so that when the communication message from the sender arrives, the message can be delayed until the receiver reaches the time.  Later in Section~\ref{subsec:comm}, we will show that such condition can be easily achieved by adjusting parameters in the clock of the receiver and putting the so-called placeholder in the event queue of the receiver. 
  
  
  Now let us look at and prove Lemma~\ref{thereom:comm} that connects Definition~\ref{def:det} and Definition~\ref{def:comm}.  

\begin{lem} \label{thereom:comm} 
A deterministic RF remains deterministic after communicating with other RFs, i.e., following Definition~\ref{def:det}, if the communication obeys Definition~\ref{def:comm}.
\end{lem}

\begin{proof}\vspace{-0.05in}
See Appendix~\ref{lemma1} for proof. \vspace{-0.05in}
\end{proof}

Now we can show our important theorem about determinism, i.e., Theorem~\ref{theorem:detdefeattiming}. 

\begin{thrm} \label{theorem:detdefeattiming} 
{\bf{(Determinism prevents timing attacks)}} If a RF with a target secret is deterministic, an adversary observer---no matter in this RF or another RF---cannot infer the target secret. 

Specifically, if the observer measures the ending and starting timestamps ($t_{End}$ and $t_{St}$) of the target secret, the following holds: $\Delta t = t_{End} - t_{St} = const$.
\end{thrm}

\begin{proof} \vspace{-0.05in}
See Appendix~\ref{theorem1} for proof.  \vspace{-0.05in}
\end{proof}

Theorem~\ref{theorem:detdefeattiming} says that the observer will always obtain the same timing information when measuring the target secret.  Specifically, determinism normalizes the target secret so that the delta between the ending and starting timestamps of the target secret, from the viewpoint of an observer, is always the same, i.e., deterministic, even if the target secret happens multiple times in different runtime environments.  In other words, we prove that {\it determinism prevents timing attacks}.





\section{Reference Frames}
\label{sec:referenceframe}


In this section, we will introduce how to make RFs in current browsers deterministic. The ideal yet simple solution is to make every RF deterministic, but such solution is impossible 
in some scenarios.  For example, a network message may contain the information about the physical clock, and there is no way to make such RF deterministic (as it belongs to the external world).  The good news is that according to Theorem~\ref{theorem:detdefeattiming}, we only need to make RFs with target secrets deterministic to prevent timing attacks.  Now let us look at several examples of RFs. 




\subsection{JavaScript Main RFs} \label{subsec:jsrf}


JavaScript main RF needs to be deterministic because the performance of JavaScript execution as shown in Figure~\ref{fig:example} is a target secret to the JavaScript itself.  In this subsection, we start from some background knowledge, present our core definition of a deterministic JavaScript RF, and then use  Figure~\ref{fig:example} as an example to show how timing attacks are prevented. 



\subsubsection{Background}


To execute a JavaScript program, a browser---particularly the JavaScript engine part---will parse and convert the JavaScript program into a special form called operation code (opcode), sometimes referred as bytecode as well. 
 Then, depended on the execution mode, i.e., interpreter or just-in-time (JIT) compilation, the opcode will be interpreted or converted to machine code for execution.  


Note that although the definition of opcodes is specific to the type of the browser, such as Firefox and Google Chrome, determinism can be associated with any set of opcodes. In our prototype implementation, we modified Firefox browser and thus used the definition of opcodes in SpiderMonkey, the Firefox JavaScript engine. 


\subsubsection{Deterministic JavaScript RF} 
We will give an alternative definition of determinism in JavaScript RF in Definition~\ref{def:jsdet}.  This alternative definition for JavaScript is similar to the special definition~\cite{deterland} that researchers used in lower level, such as virtual machines.  We will then use Theorem~\ref{th:deteq} to show that this definition is equivalent to our general determinism definition of RFs (Definition~\ref{def:det}) in the context of JavaScript engine. 

\begin{defi} \label{def:jsdet} 
{\bf (Deterministic JavaScript RF)}  Given a set of opcodes ($SO$) generated from a program, a fixed initial state, and a set of fixed inputs, a JavaScript RF is deterministic if and only if the followings hold for any two executions ($E^1$ and $E^2$) on different runtime environments:


For every opcode ($op\in SO$) in $E^1$ and $E^2$, $t_{op}^1 - t_{op}^2=C$, where $t_{op}^k$ is the timestamp when the opcode $op$ is executed in $E^k$, and $C$ is a constant related to only the starting time of $E^1$ and $E^2$.


\end{defi}

\begin{thrm} \label{th:deteq}
{\bf (Determinism Definition Equivalence)} In the context of JavaScript engine, Definition~\ref{def:jsdet} is equivalent to Definition~\ref{def:det}.
\end{thrm}

\begin{proof} \vspace{-0.05in}
See Appendix~\ref{app:deteq} for proof.  \vspace{-0.05in}
\end{proof}


Now let us look at Definition~\ref{def:jsdet}, which says that JavaScript operations in different executions of a deterministic RF will follow a determined distribution pattern over time axis when inputs are fixed and span over time in a fixed pattern.  That is, if one makes a translation from one execution to the starting point of another, the two executions look the same.



There are two methods to fix the execution pattern for JavaScript.  First,
%
%
 one intuitive method is to still use physical clock in this RF but arrange the opcode execution sequence following a  pre-determined pattern.  However, such arrangement will slow down a faster execution, and cannot make up for a slower execution (as the clock does not wait for a slower execution).

 Second, what is being used in this paper is that instead of changing the execution, we change the clock in this RF so that the perceived execution pattern from the viewpoint of the observer, i.e., the JavaScript program, is fixed, although we, such as oracle users knowing the physical time, see the execution pattern differently.  
 

\subsubsection{An Example of Deterministic JavaScript RF}
 Given Definition~\ref{def:jsdet}, there are many possible clock definitions in deterministic JavaScript RFs, which have no differences from the perspective of preventing timing attacks. 
 Now, considering simplicity and performance overhead, we define a specific clock used in our implementation, which ticks based on the following criteria:









($i$) When there are opcodes running, the clock ticks with regards to the executed opcode.  That is, we have the following equation, 
$t_{now} = t_{start} + \sum_{op\in EJO}{unit_{op}}$, 
where $EJO$ is the set of executed JavaScript opcodes and $unit_{op}$ is the atomic elapsed time for that opcode. For simplicity, if we normalize $unit_{op}$ as $unit$, we will have 
 Equation~\ref{eq2}. 
\begin{equation} 
t_{now} = t_{start} + count_{EJO}\times unit \label{eq2} 
\end{equation}

($ii$) When there are no opcodes running, the clock can tick based on any criterion, because there is no target secret in the JavaScript RF---as JavaScript is not executing---and more importantly no observer in the JavaScript RF measuring target secrets. 
 One can consider a JavaScript RF in this state as a blackbox that no observers know what happens inside (thinking about Schr\"{o}dinger's cat). 
 The key for such state is that when another RF communicates with the JavaScript RF, we need to synchronize clocks in both RFs (Definition~\ref{def:comm}).  The reason is that when there exists communication, e.g., when a callback is invoked, an observer will be present.  We will discuss more details about the communication in Section~\ref{subsec:comm}.  




The aforementioned clock makes JavaScript RF deterministic.  
  Assume that we have two executions of the same JavaScript at $t^1_{start}$ and $t^2_{start}$ separately.  For an arbitrary opcode $op$, we will have $t^1_{op}-t^2_{op}=(t^1_{start} +  count_{EJO}\times unit) - (t^2_{start} +  count_{EJO}\times unit)  = t^1_{start} - t^2_{start} = C \label{eq4}$, thus satisfying Definition~\ref{def:jsdet}.

\subsubsection{Preventing the Synchronous Attack in Figure~\ref{fig:example}}  \label{subsec:defeatsyncattack}

Now let us use Figure~\ref{fig:example} to discuss how the synchronous attack (Line 1--21) is prevented in JavaScript main RF.  For the simplicity of explanation, we assume that one JavaScript statement will form into one opcode.  In this case, the return value of the $nextEdge$ function (Line 1--6) is always 1, because for each JavaScript statement (Line 4--6), the clock of the JavaScript RF will tick forward by $unit_{op}$, the unit time for an opcode.  Then, $start$ is not equal to $performance.now()$ at Line 6, and the function just returns 1 at Line 7.  For JavaScript fingerprinting, the $fingerprint$ is fixed as 1 (Line 11).  For a side-channel attack, there is only one opcode between Line 16 and Line 18. Therefore, $stop-start$ equals $unit_{op}$, and $fingerprint-remain$ equals 0, which makes $duration$ as a constant number, i.e., $unit_{op}$.

\subsection{Auxiliary RFs} \label{subsec:aurf}

In this part of the section, we discuss how to make auxiliary RFs deterministic. Similar to JavaScript main RF, an auxiliary RF  needs to be deterministic if it has a target secret. 
%
%
%
%
%
 Now let us look at several examples of auxiliary RFs.  Note that this is an illustration but not a complete list.  Just as what Tor Browser~\cite{torbrowser} and Fermata~\cite{fuzzyfox} did to add noise at various places, we also need to investigate the same places (RFs) in browser and make them deterministic.  



\subsubsection{DOM Auxiliary RF}  A DOM auxiliary RF is attached to a JavaScript main RF: Such RF is created by 
the invocation of DOM operation via JavaScript and destroyed when the DOM operation finishes.  
 We need to make DOM auxiliary RF deterministic, because the execution time of a DOM operation can be used to infer the size of the resource involved in the operation~\cite{VanGoethem:2015:CST:2810103.2813632}.  The clock in a DOM auxiliary RF will inherit from the JavaScript one, when the DOM RF is created by the JavaScript. Then, the clock in the DOM auxiliary RF will tick for a constant time, i.e., when the DOM operation returns---no matter synchronously via a function call or asynchronously via an event, the clock is incremented by a constant time.  Note that when there is no communication between DOM and JavaScript, because there are no observers, the clock in the DOM RF is in a uncertain state. 




\subsubsection{Networking Auxiliary RF}  A networking auxiliary RF is also attached to the JavaScript main RF.  Networking auxiliary RFs are created by an HTTP request, and destroyed by an HTTP response, both initiated from the JavaScript main RF. 
 Networking auxiliary RF is 
 deterministic for cross-origin requests, but remains non-deterministic for same-origin ones.  The reason is as follows.  If a request is from cross origin, the initiator, i.e., the JavaScript, can infer the size of the response based on the time of processing the response~\cite{VanGoethem:2015:CST:2810103.2813632}.  That is, the networking RF contains a target secret.  Specifically, the clock in such RF always ticks for a constant time between the request and the response.   As a comparison, if a request is from the same origin as the JavaScript, the JavaScript by any means has access to the response, i.e., there is no target secret in the RF.  Because an external server may embed the physical time in the response, we will let the clock in same-origin networking RF tick based on the physical time.  Note that the communication between same-origin networking and JavaScript RFs is still deterministic per Definition~\ref{def:comm} so that the JavaScript RF cannot obtain the physical time.  



\subsubsection{Video Auxiliary RF}  Video auxiliary RF is created when a JavaScript renders a video on an HTML5 canvas as an animating texture. According to  Kohlbrenner and Shacham~\cite{fuzzyfox}, such video rendering can be used as an implicit clock via $requestAnimationFrame$, because most modern browsers render video in a 60Hz frequency.   That is, one can infer the physical time based on the current displayed video frame.   We do not need to make a video RF deterministic, i.e., the clock in such RF ticks based on physical time, because it contains no target secrets.  Note again that the communication between the video auxiliary and other deterministic RFs is still deterministic.  

\subsubsection{WebSpeech or WebVTT Auxiliary RF} WebSpeech and WebVTT are another two venues of implicit clocks in the modern browsers as mentioned by Kohlbrenner and Shacham~\cite{fuzzyfox}.  A WebSpeech RF with a $SpeechSynthesisUtterance$ object is created by a $speak()$ method, and destroyed by a $cancel()$ method which fires a callback with a high resolution duration of the speech length.  A WebVTT is a subtitle API that can specify the time for a specific subtitle and check for the currently displayed subtitle. Similar to the video auxiliary RF, because there are no target secrets, both RFs are non-deterministic and the clocks in these two RFs tick based on the physical time. 










\section{Communication between RFs} \label{subsec:comm}


In this section, we discuss, when one RF communicates with another, how to synchronize clocks according to  Definition~\ref{def:comm}. 
 Before explaining our technique, i.e., a priority queue, we first need to understand how the RF communication 
 works in legacy browsers. 

The RF communication is handled by a so-called event loop~\cite{eventloop}, in which a loop keeps fetching events from a queue structure called {\it event queue}.  
 Let us again use $setTimeout$ to explain the RF communication with the {\it event queue}.  When the main RF calls $setTimeout$, the main RF will dispatch an event into the event queue of the timer RF, i.e., launching a communication from the main RF to the timer RF.  The timer RF will fetch the event, process it (i.e., waiting in this example), and then dispatch another event, called a callback, back to inform the main RF the completion of $setTimeout$, i.e., launching a communication from the timer RF to the main RF. 
 
 
In our paper, we replace the event queue in the main RF with a priority queue. 
Specifically, we reserve a callback place in the queue for events that the main RF dispatched to other RFs.  When the callback is dispatched back to the main RF, the callback is synchronized at the reserved place following Definition~\ref{def:comm}.  Take a deterministic auxiliary RF for example.  The duration of processing an event in such RF always equals to a fixed value from the perspective of the main RF observer.
  To achieve that, the main RF pre-assigns an expected delivery time for every callback as a priority in the queue so that the callback will be only delivered at that specific, fixed time from the viewpoint of the main RF observer. One important property of this priority queue is that 
 when multiple events happen in many deterministic RFs, the queue will arrange all the callbacks in a pre-determined, fixed sequence.  Now let us look at the details. 

\subsection{Priority Queue} \label{subsec:pqueue}


\begin{figure}
\begin{center}
\centering
\includegraphics[width=0.8\columnwidth]{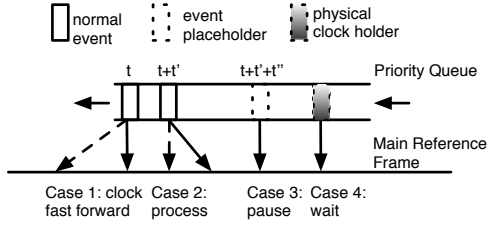} 
\vspace{-0.2in}
\caption{Mechanism of Priority Queue in the Main RF. }
\label{fig:queue}
\vspace{-0.25in}
\end{center}
\end{figure}


Now let us look at how the main RF interacts with the priority queue.  Similar to the current event model, when the main RF is idle, it will try to fetch an event from the queue.  In our model, the event with the highest priority, i.e., the one with the earliest delivery time, will be returned to and processed by the main RF.  
 There are two sub-scenarios (Case 1\&2 in Figure~\ref{fig:queue}) here.  First, when the expected delivery time is ahead of the clock in the main RF (called the main clock), the browser kernel will move the main clock to the expected time, and let the main RF process the event.  Second, when the expected delivery time is behind the clock in the main RF (i.e., the main RF was processing other events at the expected delivery time), the main RF will directly process the event without changing the clock. 

Note that both cases follow the deterministic definition of communication in Definition~\ref{def:comm}.  
 For Case 1, because two clocks are synchronized at the delivery time, the communication is deterministic. 
  For Case 2, we can 
 consider that the event is delivered at the expected delivery time following Definition~\ref{def:comm}.  However, because JavaScript is a single-thread language, the main RF cannot process the event when busy.  That is, the main RF has to postpone the processing of the event to the time when it is idle.  Thus we can combine the delivery and the processing, because even if the event is delivered, the main RF cannot notice and process the event. 

Besides normal events, there are two special events (Figure~\ref{fig:queue}): an event placeholder and a physical clock holder.
 An event placeholder is a virtual slot indicating that a real event should be at this place in the queue, but the event has not been finished in the auxiliary RF. When the main RF fetches an event placeholder (Case 3), the main RF needs to pause its own clock and wait until the event arrives or is canceled (e.g., timeout).  If the event arrives, it will replace the event holder, and the main RF will process the event; if the event is canceled, the main RF will process the next event in the queue.  

A physical clock holder is a mapping from the current physical time to the main RF clock time. Because the physical clock keeps ticking, the place of the holder also changes, i.e., the priority of the holder will be lowered constantly.  The usage of such holder is two-fold. First, when an auxiliary RF with the physical time tries to communicate with the main RF, the browser kernel will add the event to the holder's place in the queue. Second, in a very rare case (Case 4), when the main RF tries to fetch the physical clock holder, i.e., the main clock ticks faster than the physical one, the browser kernel will pause the main RF until the main clock is behind the physical one.  The reason is that an event with physical time can come at any time, and if the main clock is faster, the browser kernel cannot synchronize the main clock with the physical one.  To prevent such scenario, we can set the $unit$ in Equation~\ref{eq2} to a very small number so that the main clock cannot catch up with the physical one in the fastest machine. 



\subsection{RF Communication} \label{subsubsec:commmainaux}
We now look at how the priority queue can be used in the communication between RFs. 

\subsubsection{Main and Auxiliary RF Communication}  Let us first look at the communication between main and auxiliary RFs.  First, say the auxiliary RF is deterministic, and Figure~\ref{fig:masync} shows how it works.  Before creating such auxiliary RF, the main RF first put a placeholder in the priority queue (Step 1).  Because the auxiliary RF is deterministic, the main RF can predict the expected delivery time.  Then, the auxiliary RF is started (Step 2) and keeps running until it finishes and tries to deliver an event to the main RF.  The event will be put in the priority queue, replacing the previous placeholder (Step 3). At the expected delivery time in the scale of the main clock, the event is delivered to the main RF (Step 4).  If the main RF is busy, the execution might be delayed (Step 5).  As mentioned, if the execution is delayed, the main RF can fetch the event at the delayed time.

Second, say the auxiliary RF is non-deterministic.  Examples of such RFs include these created by the main RF (e.g., same-origin networking RF) and these created by the user (e.g., mouse or keyboard RF).  
 When such RF is created, because we cannot predict the expected delivery time, no placeholder is created.   If the auxiliary RF tries to communicate with the main RF, i.e., an event is fired, the event will be put at the place of the physical clock holder in the priority queue.   




\begin{figure}
\begin{center}
\centering
\includegraphics[width=\columnwidth]{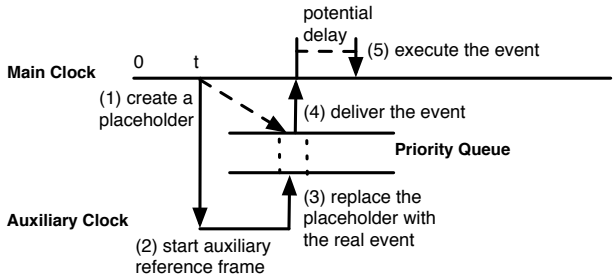} 
\vspace{-0.4in}
\caption{Communication between Main and Auxiliary RFs. }
\label{fig:masync}
\vspace{-0.25in}
\end{center}
\end{figure}





\subsubsection{Communications between Two Main RFs}
 We then discuss how two main RFs communicate with each other using the priority queue, which follows Definition~\ref{def:comm}.  
 There are two scenarios in this communication: (i) the sender's clock is ahead of the receiver, and (ii) the sender's clock is behind.  First, when the sender's clock is ahead, the browser kernel can put the communication message, i.e., an event, in the priority queue at the position of a future time. Then, when the event is executed, two clocks are synchronized---everything is the same as the communication between a main and auxiliary RF, and follows the first condition in Definition~\ref{def:comm}. 

Second, when the sender's clock is behind, the browser kernel will put the communication event in the front of the priority queue, i.e., deliver the event immediately.  In such case, the receiver knows two clocks, i.e., the sender's and its own.  Let us first consider that the sender has a target secret, i.e., the sender's RF is deterministic.  In such case, the communication follows the second condition in Definition~\ref{def:comm}.


Let us then consider that the sender's RF contains no target secrets and is not deterministic.  An example is that the main RF is not running JavaScript and just gets synchronized with an auxiliary RF with a physical clock.  
 Such case contradicts with our assumption in the beginning that the sender's clock is behind the receiver's, because as shown in Section~\ref{subsec:pqueue}, relying on the physical clock holder, the clock in a deterministic main RF, i.e., the one in the receiver, is always behind the physical clock. 
  Therefore, the scenario boils down to our first case, where the sender's message is put in the priority queue, and both clocks are synchronized.

\subsection{Preventing the Asynchronous Attack in Figure~\ref{fig:example}}  \label{subsec:defeatasyncattack}

Now let us look at how the priority queue can help to prevent the asynchronous attack (Line 22--28) in Figure~\ref{fig:example}.  In one sentence, when the adversary measures the duration of the asynchronous target secret using the reference clock (i.e., $setInterval$ in this example), no matter how long the target secret takes to finish, the adversary will always obtain a fixed, deterministic value.

Here are the details. 
 In the $sideChannelAsync$ function, when the asynchronous target secret is invoked in function $targetSecretAsync$ (Line 26) say at $t_{init}$, the browser kernel will put an event placeholder in the priority queue with an expected delivery time, say at $t_d$, and create a deterministic RF to execute the asynchronous secret. That is, $t_d - t_{init}$ is a constant due to determinism.  
 Next when the $setInterval$ function is invoked, the browser kernel will create another RF for waiting and put an event placeholder in the priority queue with an expected delivery time at $t_{init}+u$ where $u$ is the variable at Line 27. Because the implicit clock introduced by $setInterval$ is to measure $targetSecretAsync$, $u$ should be much smaller than $t_d - t_{init}$.  That is, the event placeholder for $setInterval$ should be ahead of the one for the target secret in the priority queue.  Therefore, function $countFunc$ (Line 23) is executed, and then another event placeholder for $setInterval$ is created.  After $countFunc$ is executed for $round((t_d-t_{init})/u)$ times, the callback for the target secret, i.e., $callback$, is executed.  Note that if the target secret finishes execution, the placeholder for the target secret is replaced by the callback, which can be executed immediately; if not, the JavaScript RF will wait there for the placeholder to be replaced by the callback.

In sum, we have three RFs in this example: JavaScript RF, DOM RF for $setTimeout$, and a deterministic RF for executing the target secret.  The duration measured by $callback$ is always $round((t_d-t_{init})/u) * u$, a deterministic value. Therefore, the asynchronous attack is prevented.










\section{Implementation}

We have implemented a prototype of deterministic browser called \sys with 1,687 lines of code by modifying Firefox nightly 51.0a1 at 40 different files.  The implementation is available at this repository (\url{https://github.com/nkdxczh/gecko-dev/tree/deterfox}).  Let first look at the implementation RFs. As mentioned, legacy Firefox already has RF implementations via OS level threads (xpcom/threads/nsThread.cpp).  Different RFs are different subclasses of nsThread, such as TimerThread (xpcom/threads/TimerThread.cpp). Now we introduce the core part of our implementation: deterministic JavaScript RF and priority queue.


 Let us first look at the deterministic JavaScript engine (js/src).  We associate a counter with each JavaScript context in SpiderMonkey, the Firefox JavaScript engine of Firefox. When $performance.now$ is invoked, the counter multiplying a unit time will be returned.  Note that some Firefox scripts, e.g., these with ``chrome://'' and ``resource://'' origins, are also running in the same context as JavaScript from the website.  We create separate counters for such Firefox scripts.  
 Here is how the counters are incremented. 
 Specifically, SpiderMonkey has three modes of executing JavaScript code, one interpreter (js/src/vm/Interpreter.cpp) and two just-in-time (JIT) compilation modes (IonMonkey at js/src/jit/Ion.cpp and BaselineJIT at js/src/jit/BaselineJIT.cpp). 
  We increment the counter in all the three modes. 
  In the interpreter mode, the counter will be incremented for each opcode; in both JIT modes, the counter will be incremented based on the compiled JavaScript block. Our current implementation does not add the counter in the compiled code directly but before and after the compiled code execution.  Note that this implementation is still deterministic, because it makes the JavaScript execution follow a certain pattern over time. 

We then change the event queue for the deterministic communication. For each thread in Firefox, other threads can dispatch a runnable to that thread and put (PutEvent method in xpcom/threads/nsThread.cpp) the runnable in the event queue (mEventsRoot object in xpcom/threads/nsThread.cpp).  Note that there are many queues in Firefox, and these queues are sometimes hierarchical, i.e., one queue may dispatch events to another queue.  This event queue that we talk about 
 is directly associated with a thread, i.e., the inner level queue.  


\section{Evaluation}


We evaluate \sys based on the following metrics:

\begin{icompact}
\item {Robustness to timing attacks.}  \hspace{0.01in} We evaluate \sys and other existing browsers including commercial ones and research prototype against existing timing attacks presented in the literature. 
\item {Performance overhead.} \hspace{0.01in} We evaluate the performance overhead of \sys from the perspective of a user of the browser, i.e., based on the physical time obtained from a standard Linux machine. 
\item {Compatibility.} \hspace{0.01in} We evaluate the compatibility of \sys using two tests: an automated browser testing framework called Mochitest, and the rendering of Top Alexa websites. 
\end{icompact}

\subsection{Robustness to Timing Attacks}

In this part of the section, we evaluate the robustness of existing browsers and \sys against timing attacks.  Formal proof about why a deterministic browser prevents timing attacks can be found in Theorem~\ref{theorem:detdefeattiming}.  In this subsection, we focus on the empirical evaluation of \sys and other prior arts against timing attacks presented in the literature.  Specifically, four existing browsers---namely Firefox, Google Chrome, FuzzyFox~\cite{fuzzyfox},\footnote{FuzzyFox is a prototype implementation of Fermata, and we use the version downloaded from their repository with the default parameters. } and Tor Browser~\cite{torbrowser}---are used for comparison, which range from commercial browsers to research prototype.  
 An overview of the evaluation can be found in Table~\ref{tab:robustness}.

\begin{table}[t]
\centering
\caption{Robustness of Five Browsers against Different Attacks. (``Clock-edge''~\cite{fuzzyfox}; ``Clock-edge-m'': an modified version of ``Clock-edge''; ``Img ($S$)'': the image loading side-channel attack~\cite{VanGoethem:2015:CST:2810103.2813632} inferring two files with $S$ size difference; ``Script ($S$)'': the script parsing side-channel attack~\cite{VanGoethem:2015:CST:2810103.2813632} inferring two files with $S$ size difference;  ``Script-implicitClock (2M)'': a modified version of ``Script (2M)'' using $setInterval$ as an implicit clock; ``Cache attack'': a side-channel attack~\cite{Oren:2015:SSP:2810103.2813708}; ``Cache-m'': a covert channel modified from ``Cache attack''; ``SVG Filtering'': the SVG filter attack from Stone~\cite{stone2013pixel}.)}
\label{tab:robustness} \vspace{-0.1in}
{\sffamily\fontsize{5.5}{9}\selectfont
\begin{tabular}{l*5c}
\toprule
 & Chrome&Firefox&Tor Browser&FuzzyFox&DeterFox\\
\midrule
Clock-edge~\cite{fuzzyfox} &$\bf \times$&$\bf \times$&$\bf \times$&$\checkmark$  & $\checkmark$  \\
Clock-edge-m &$\bf \times$&$\bf \times$&$\bf \times$&$\bf \times$  & $\checkmark$  \\
Img (100K)~\cite{VanGoethem:2015:CST:2810103.2813632} &$\bf \times$&$\bf \times$&$\checkmark$&$\checkmark$  & $\checkmark$  \\
Img (5M)~\cite{VanGoethem:2015:CST:2810103.2813632} &$\bf \times$&$\bf \times$&$\bf \times$&$\bf \times$  & $\checkmark$  \\
Script (100K)~\cite{VanGoethem:2015:CST:2810103.2813632} &$\bf \times$&$\bf \times$&$\checkmark$&$\checkmark$  & $\checkmark$  \\
Script (2M)~\cite{VanGoethem:2015:CST:2810103.2813632} &$\bf \times$&$\bf \times$&$\bf \times$&$\bf \times$  & $\checkmark$  \\
Script-implicitClock (2M) &$\bf \times$&$\bf \times$&$\bf \times$&$\bf \times$  & $\checkmark$  \\
Cache attack~\cite{Oren:2015:SSP:2810103.2813708} &$\bf \times$&$\bf \times$&$\checkmark$&$\checkmark$  & $\checkmark$  \\
Cache-m &$\bf \times$&$\bf \times$&$\bf \times$&$\bf \times$  & $\checkmark$  \\
SVG Filtering~\cite{stone2013pixel} &$\bf \times$&$\bf \times$&$\bf \times$&$\bf \times$  & $\checkmark$  \\
\bottomrule
\end{tabular}

\scriptsize
Note: $\checkmark$ means the browser is robust to the attack, and $\bf \times$ not.
}
\vspace{-0.10in}
\end{table}

\subsubsection{Clock-edge Attack}
 A synchronous version of the clock-edge attack~\cite{fuzzyfox} is presented in Section~\ref{sec:threatmodel}.  Apart from that version, we also design another version specifically targeting FuzzyFox.  The modified version runs an operation for a considerable long time and calculate the difference between two $performance.now$ asynchronously. 
  The attack is repeated for ten times to average out the jitter added by FuzzyFox. 
  Additionally, because two $performance.now$ are obtained asynchronously, they do not equal to each other according to the FuzzyFox paper.

The first two rows of Table~\ref{tab:robustness} shows the results of both the original and modified version of the clock-edge attack.  All browsers except \sys and FuzzyFox are vulnerable to the original version, i.e, the minor clock can be used to measure a target secret with $\mu s$ accuracy.  Both \sys and FuzzyFox will show a constant time for the minor clock, making it unusable for measurement of target secrets.  For the modified version, all browsers except \sys are vulnerable.  We also test five browsers on the clock-edge attack for the fingerprinting purpose (see Section~\ref{sec:threatmodel}), and the results are the same as the one for the measurement purpose.  

\begin{figure}
\begin{center}
\centering
\includegraphics[width=0.85\columnwidth]{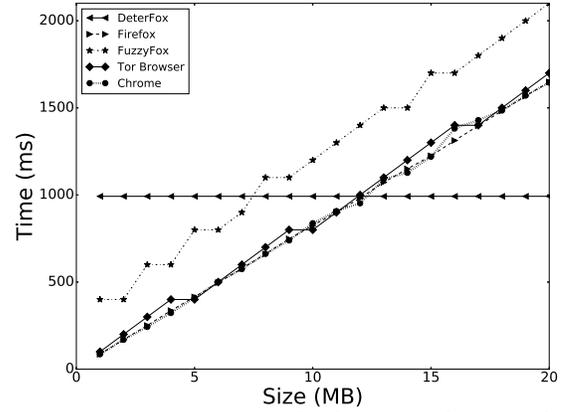} 
\vspace{-0.15in}
\caption{Script Parsing Attacks from Goethem et al.~\cite{VanGoethem:2015:CST:2810103.2813632} (We change the size of parsed scripts in the attack, and measure the time to trigger an ``onerror'' event. Each point in the graph is the median value of nine repeated tests. Note that all timestamps are obtained via JavaScript.)}
\label{fig:jsoverhead}
\vspace{-0.15in}
\end{center}
\end{figure}

\subsubsection{Side-channel Attacks from Goethem et al~\cite{VanGoethem:2015:CST:2810103.2813632}}  We evaluate the robustness of browsers against two side-channel attacks from Goethem et al.~\cite{VanGoethem:2015:CST:2810103.2813632}, namely script parsing and image loading.  Both attacks try to load a file---which is not image or script and from a different origin---and measure the time between the start and when the ``onerror'' event is triggered.  Details can be found in their paper. 

In the evaluation, we change the file size from 1MB to 20MB and observe the loading or parsing time.  The experiment for each file size is repeated for 15 times, and we discard the first six results, because the browser is sometimes busy loading system files during startup, which affect the result.  Then, we obtain the median value of the rest nine results for each file size and show the result in Figure~\ref{fig:jsoverhead} and~\ref{fig:imgoverhead}.  Note that the file size we use is larger than the one in Goethem et al.~\cite{VanGoethem:2015:CST:2810103.2813632}, because both FuzzyFox and Tor Browser add noises to the loading or parsing time so that it is hard to differentiate files with small sizes.  However, the file size that we use is still reasonable for normal web communications, such as for email attachment.  Now let us look at the results for script parsing and image loading separately. 


Figure~\ref{fig:jsoverhead} shows the results of script parsing attack for five browsers.  First, the parsing time for \sys is a constant, deterministic number, i.e., \sys is entirely robust to the script parsing attack. 
 Second, the parsing time for both legacy Firefox and Google Chrome is linear to the file size, confirming the results reported by Goethem et al.  Lastly and more importantly, the parsing time for both Tor Browser and FuzzyFox is a stair step curve with regards to the file size.  That is, both browsers can defend against such parsing attack in small scale when the file sizes differ a little, but fail when the file sizes differ a lot, e.g., with 2MB differences (Table~\ref{tab:robustness}). 


\begin{figure}
\begin{center}
\centering
\includegraphics[width=0.85\columnwidth]{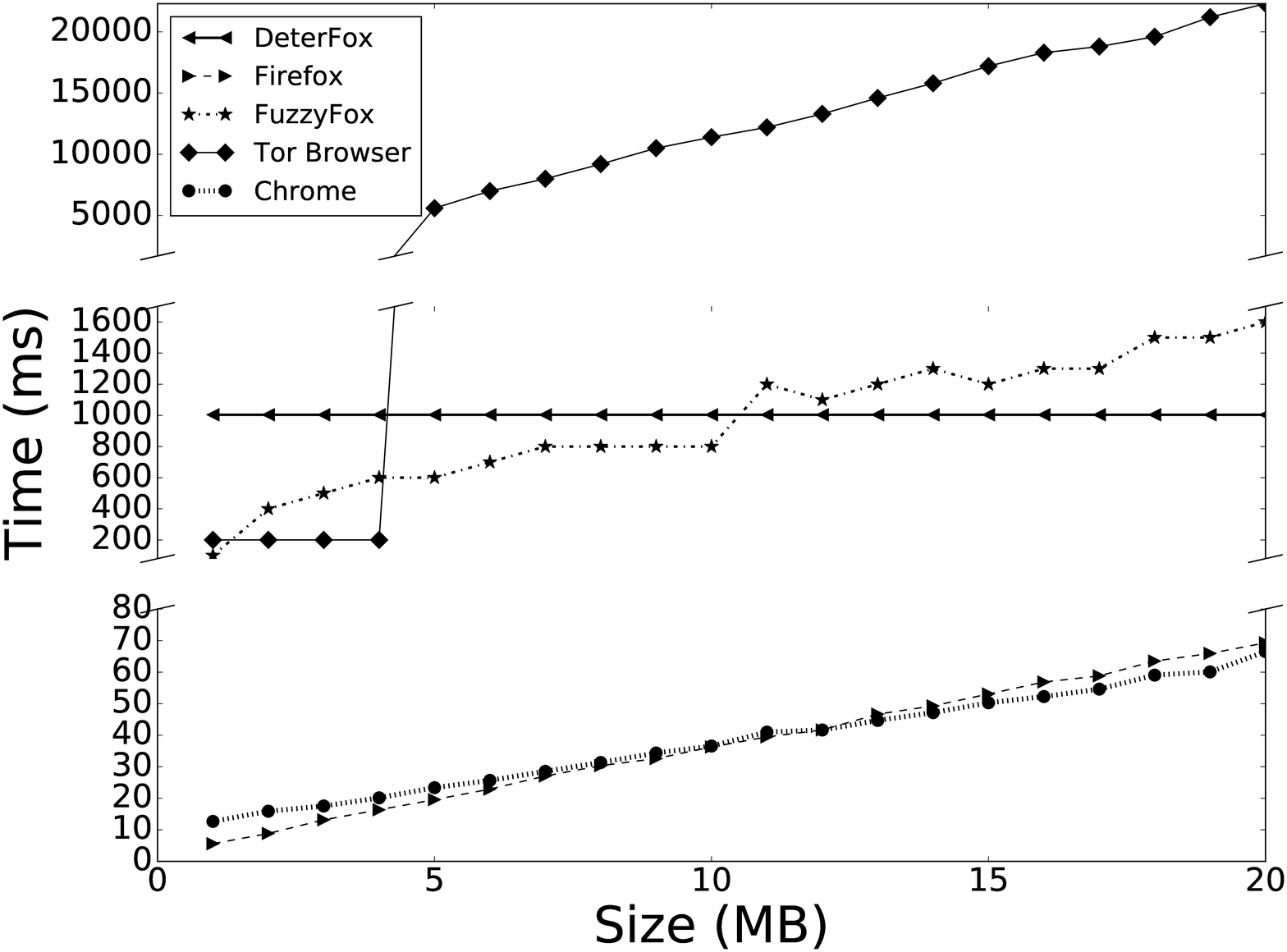} 
\vspace{-0.15in}
\caption{Image Loading Attacks from Goethem et al.~\cite{VanGoethem:2015:CST:2810103.2813632} (We change the size of parsed scripts in the attack, and measure the time to trigger an ``onerror'' event. Each point in the graph is the median value of nine repeated tests. Note that all timestamps are obtained via JavaScript.)}
\label{fig:imgoverhead}
\vspace{-0.15in}
\end{center}
\end{figure}

Figure~\ref{fig:imgoverhead} shows the results of image loading attack for five browsers.  First, similarly the loading time for \sys is a constant, showing that \sys is robust to such image loading attack, and both legacy Firefox and Google Chrome are vulnerable to the attack because the loading time increases as the file size linearly.  Second, the parsing time for FuzzyFox is still a stair step curve with a few fluctuations.  That is, FuzzyFox is also vulnerable to such attack if the file size differs much, e.g., with 5MB differences (Table~\ref{tab:robustness}).  Lastly, the loading time for Tor Browser is interesting, because the time stays as a constant number when the file size is below 4MB, but increases linearly when the file size is above 5MB.  That is, Tor Browser is also vulnerable to the attack when loading a large file. 

Note that in addition to using $performance.now$ to measure the loading time, following the asynchronous attack in Figure~\ref{fig:example}, we also tried $setInterval$ as an implicit clock for the script parsing attack with 2M file size difference.  The result (``Script-implicitClock (2M)'' line in Table~\ref{tab:robustness}) shows that only \sys can defend against the attack with the $setInterval$ implicit clock---this is expected as only \sys can defend against the original script parsing attack with 2M file size difference. 

\subsubsection{Cache Attacks from Oren et al~\cite{Oren:2015:SSP:2810103.2813708}} Oren et al.~\cite{Oren:2015:SSP:2810103.2813708} propose a practical cache-based side-channel attack using JavaScript, which can correctly recognize the website running in another tab by creating a memory access pattern with zeros and ones called memorygrams. We re-implement their attack 
 and test all five browsers against the attack.  As shown in the seventh row of Table~\ref{tab:robustness}, both Chrome and Firefox are still vulnerable, but Tor Browser, FuzzyFox, and \sys are robust, i.e., the memorygrams in these three browsers are all zero.  Note that although Tor Browser was vulnerable,
  people have patched Tor Browser  to defend against this attack by further reducing the time resolution. 

Because both Tor Browser and FuzzyFox only limit the bandwidth of the side/covert channel,
 we create a modified version to show that a covert channel is still possible (the eighth row of Table~\ref{tab:robustness}).  Specifically, we create two iframes from different origins that talk to each other based on the cache attack for a long time.  
 The results show that if the attack lasts for two minutes, the receiver in both Tor Browser and FuzzyFox can successfully differentiate the one-bit information. That is, 
  the channel still exists with significantly reduced bandwidth in both Tor Browser and FuzzyFox.  

\subsubsection{SVG Filtering Attacks from Stone~\cite{stone2013pixel}} Stone shows that the performance of SVG filter can be used to differentiate the contents of images.  We apply the feMorphology SVG filter mentioned in the Stone's paper~\cite{stone2013pixel} to perform an erode operation upon two images, one with size 1920$\times$1080 and another with size 200$\times$200, and measure the operation time using $requestAnimationFrame$. Both images are generated randomly with a mixture of RGB colors and we perform the attack for 20 times. 


\begin{table}[t]
\centering
\caption{Measured Time to Perform an SVG Filter (All the numbers are in $ms$ and averaged from 20 experiments).}
\label{tab:svg} \vspace{-0.1in}
{\sffamily\fontsize{5.5}{9}\selectfont
\begin{tabular}{l*5c}
\toprule
 & Chrome&Firefox&Tor Browser&FuzzyFox&DeterFox\\
\midrule
200$\times$200 image & 16.66 & 17.01 & 18.94 & 109.09 & 192 \\
1920$\times$1080 image & 17.87 & 50.35 &  106.67 & 114.54 & 192 \\
\bottomrule
\end{tabular}}
\vspace{-0.15in}
\end{table}

The last row of Table~\ref{tab:robustness} shows the summary results of the attack.  All other browsers except \sys are vulnerable, i.e., the measured operation time for the large (1920$\times$1080) image is constantly longer than  the small one (200$\times$200) as shown in Table~\ref{tab:svg}.  
 Based on the results, an adversary can differentiate whether an cross-origin image is a thumbnail or full-sized. By contrast, \sys can defend against  the attack---i.e., no matter how large the image is, the operation time is always the same from the viewpoint of the JavaScript observer.  (The user can still tell the difference and will not experience the slowdown as in FuzzyFox and Tor Browser.)

\subsection{Performance Overhead}

\begin{figure*}
\begin{center}
\centering
\includegraphics[width=2\columnwidth]{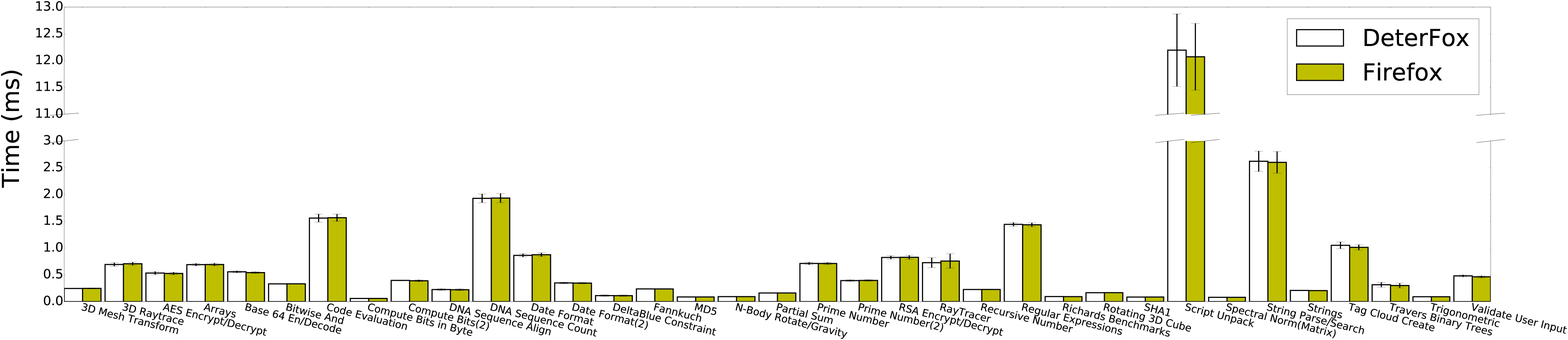} 
\vspace{-0.15in}
\caption{JavaScript Performance Benchmark Evaluation. (The median overhead of \sys is 0.63\%, and note that all timestamps are obtained based on a standard Linux machine.) }
\label{fig:jsboverhead}
\vspace{-0.10in}
\end{center}
\end{figure*}

We first evaluate the performance of \sys by JavaScript benchmark and Top 100 Alexa websites.   Then, we discuss the evaluation results. 

\subsubsection{JavaScript Benchmark} \label{subsec:jsbm}
 We first evaluate the performance overhead of \sys using a JavaScript performance test suite called Dromaeo~\cite{dromaeo}, a unified benchmark from Mozilla.  
  Each test in Dromaeo is executed for at least five times and maybe up to ten times if a significant level of error cannot be reached.  Next, the results are fit into a t-distribution to calculate the 95\% confidence interval.  Note that because we changed the clock in \sys, all the timestamps used in the experiment are based on the local time of a Linux machine instead of the JavaScript time. 

Figure~\ref{fig:jsboverhead} shows the evaluation results with the median overhead of \sys as 0.63\% and the maximum as 1\% when compared with the legacy Firefox of the same version. That is, the overhead introduced by \sys is almost ignorable in terms of JavaScript performance when observed by an oracle.  In three of the test cases, \sys is even slightly faster than Firefox because \sys does not invoke any system calls during $performance.now$, but legacy FireFox does.  If the measured time is very small, the invocation of $performance.now$ will influence the result. 
%
 Another interesting experiment is to disable the JIT mode for both \sys and the legacy Firefox and test the performance overhead to show the influence of the added counter.  The results show the median overhead for the interpreter mode as  6.1\%. 


\subsubsection{Top 100 Alexa Websites}
 We then evaluate the performance overhead of \sys against Top 100 Alexa websites. Specifically, we measure the loading time of Top 100 websites in four browsers and show the cumulative distribution function (CDF) of the loading time.  

Figure~\ref{fig:loadingtimealexa} shows the results of this experiment.  The CDF curve of \sys is very close to the one of legacy Firefox.  The median overhead of \sys compared with legacy Firefox is 0.1\%, i.e., with no statistical difference.  At contrast, both Tor Browser and FuzzyFox incur a non-ignorable overhead.  The reason is that both Tor Browser and FuzzyFox add jitters causing a delays during page loading, but \sys only adjusts the time in each RF, i.e., the time that JavaScript observes, while the user does not observe similar overhead as JavaScript does.  

\begin{figure}
\begin{center}
\centering
\vspace{-0.15in}
\includegraphics[width=0.85\columnwidth]{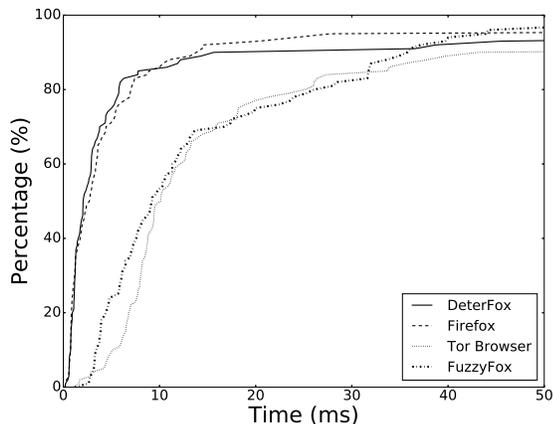} 
\vspace{-0.15in}
\caption{Cumulative Distribution Function of Loading Time of  Top 100 Alexa Websites. (Note that all timestamps are obtained based on a standard Linux machine.  The median overhead is 0.1\%, i.e., there is no statistical difference between Firefox and \sys.  The x-axis is cut off at 50ms, because some Top 100 Alexa websites are located in China and the loading time is very long, e.g., 1--2mins, and more importantly the loading time of such websites are dominated by the network latency, i.e., unrelated to the browser. )}
\label{fig:loadingtimealexa}
\vspace{-0.15in}
\end{center}
\end{figure}

\subsubsection{Discussion}
 The overhead of \sys comes from two parts: the counter and the priority queue.  The counter brings overhead because the incrementing behavior takes time. Further, the overhead is small, because the incrementing behavior is very cheap.  The priority queue does not bring any overhead if there is only one event waiting.  When there are two (or more than two) events  in the queue, if the event---which \sys arranges behind the other---comes first, the arrangement will cause overhead.  However, in practice, as shown in Top 100 Alexa websites, this case is rare, and even if it happens the other event will also come in a short time, which bring little overhead for \sys.  



\subsection{Compatibility}

We evaluate the compatibility of \sys from three perspectives: Mochitest, rendering Top 100 Alexa websites, and behind-login functionalities of popular websites. 

\subsubsection{Mochitest} Mochitest~\cite{mochitest} is a comprehensive, automated testing framework from Mozilla, and Mozilla will use Mochitest to test Firefox before each release.  The entire Mochitest that we use from Firefox nightly 51.0a1 contains 878,556 individual tests grouped into 41,264 categories, and will run for 8 hours on a standard Linux virtual machine~\cite{firefoxvm} downloaded from Mozilla official website. Note that even the legacy Firefox cannot pass all the tests on that virtual machine: Specifically, the legacy Firefox nightly 51.0a1 only passes 97.8\% of tests---859,283 out of 878,556 as passed, 9,943 as OK, i.e., in their todo list, and 9,330 as failed. 

Our evaluation results show that \sys passes 97.6\% of tests, i.e., 10,339 as failed, 10,358 as OK, and the rest as passed.  We look at these failed cases and find that \sys fails additional 1,502 tests when compared with the legacy Firefox.  Interesting, there are 493 tests that \sys passes but the legacy Firefox fails.  One possible reason is that \sys changes the sequence of events, which may fix some concurrency bugs caused by rendering sequence.  We look into some of these cases that \sys fails, and find that they belong to many categories, such as DOM and layout. 

\subsubsection{Top 100 Alexa Websites} We evaluate the compatibility of \sys by comparing the screenshot of \sys with the one of legacy Firefox browser.  Because modern websites contain many ads, even two sequential visits to the same website may render differently. Thus we manually look at these websites in which the similarity of the screenshots between \sys and Firefox is smaller than 0.9 (A similar threshold to test browser compatibility is also used in the literature~\cite{cspautogen}). 

In total, 63 websites pass the initial test and we manually examine the rest 37.  All the differences are highlighted automatically by a red circle for manual inspection.  The results show that all differences are caused by dynamic contents, such as a different ad and a news update. In sum, \sys is compatible with Top 100 Alexa websites, i.e., does not cause noticeable differences when rendering the front page of these websites. 

\subsubsection{Behind-login Functionalities of Some Popular Websites}  In this subsection, we evaluate the behind-login functionalities of several popular websites. Specifically, we choose the most popular website, according to Alexa, from several website categories including email, social network, online shopping, video, and JavaScript game.  Our manual inspection shows that we can successfully perform corresponding actions.  Here are the details. 

\begin{icompact}
\item {Email: Gmail. }  We register a new account with Google, log into Gmail, and then send an email with an attachment to another email address.  From a different computer, we reply to this email and attach another file in the reply.  Then, from the first computer, we receive the reply, look at the contents, and then download the attachment in the reply.
\item {Social Network:  Facebook.} We register a new account with Facebook, and log in.  Then, we configure several privacy settings following Facebook's tutorial, and add several friends.  Next, we post a status in Facebook by sharing a news.  We also talk with a friend via Facebook message. 
\item {Online Shopping: Amazon. } We register a new account and log in.  Then, we browse several items using the search functions in Amazon.  Next, we add a book to our shopping cart, proceed to checkout, and purchase the book with a newly added credit card. 
\item {Video: Youtube.} We log into Youtube with the Google account, and search several keywords.  Then, we select a video, and watch it for one minute.  We also post a comment under the video, and then delete the comment.  
\item {JavaScript Game.} We search the keyword ``JavaScript Game" on Google and click the first JavaScript game in the list, i.e., the fifth item during our search.  The JavaScript game is very similar to gluttonous snake in which the gamer can control an item to eat others.  We play the game for one minute. 
\end{icompact}


\section{Discussion}
 \label{sec:diss}
We discuss several problems in this section.  \vspace{0.05in}

\noindent{\it Access to Physical Time.} One common misunderstanding of our deterministic browser is that JavaScript will lose access to the physical time.  A deterministic browser only prevents JavaScript (an observer) from accessing the physical clock when there is a secret event running at the same time, because necessary conditions of a timing attack involve an observer and a secret event (See Section~\ref{sec:threatmodel}); at contrast, if there is no secret event execution, JavaScript has free access to the physical clock.  Specifically, the deterministic browser kernel will fast forward the main clock in JavaScript to the physical one as discussed in Section~\ref{subsec:comm}.  In practice, because JavaScript is an event-driven language, if JavaScript is just waiting for one single event with physical time, e.g., a network request from the same origin, the main clock will be the same as the physical clock. 




\vspace{0.05in}

\noindent{\it User-related Timing Attacks.}  As discussed in Section~\ref{sec:threatmodel}, our deterministic browser, especially the current prototype \sys, does not consider timing-related biometric fingerprinting, i.e., the secret event in our prototype is restricted to the browser itself.  To defend such biometric fingerprinting that is related to time, we need to introduce determinism into corresponding reference frames.  For example, the clock in a keyboard reference frame needs to tick based on the number of pressed keys instead of the physical clock.  This, however, is considered as our future work. 

\vspace{0.05in}

\noindent{\it External Timers and Observers.} In real-world, external timers and observers, e.g., these that reside outside a browser, are impractical as pointed out by the FuzzyFox~\cite{fuzzyfox} paper.  Any timestamps or observations provided via such methods are too coarse-grained to perform any meaningful attacks.  Therefore, neither FuzzyFox nor Tor Browser considers external timers and observers in their thread model.  Theoretically, a deterministic browser can defend against such attacks performed by external timers or observers, because according to Definition~\ref{def:comm}, two RFs are synchronized during communication.  That is, when an external observer obtains the time or an internal observer obtains an external time, the time is out-of-dated, which only reflects the synchronized time.






\section{Related Work}

In the related work section, we first discuss existing timing attacks, and then present prior work that mitigates such timing attacks, especially browser-related ones.

\subsection{Browser-related Timing Attacks}

We discuss existing browser-related timing attacks below. 

\subsubsection{JavaScript Performance Fingerprinting} 
 Mowery et al.~\cite{mowery2011fingerprinting} and Mulazzani et al.~\cite{mulazzani2013fast} show that the performance of JavaScript, i.e., how long a certain set of JavaScript code takes to execute, can be used to differentiate, or called fingerprint, different types and versions of web browsers.  The reason behind such JavaScript performance fingerprinting is that different browsers have different JavaScript engine 
   implementations and thus different runtime performance behaviors. 
 

\subsubsection{Timing-based Side or Covert Channels} 
 Timing attacks~\cite{Zhang:2011:HCD:2006077.2006774,Zhang:2012:CSC:2382196.2382230,Hund:2013:PTS:2497621.2498111,Kocher:1996:TAI:646761.706156,conf/esorics/LiuGASSK09} in general have been studied for a long time.  We focus on browser-related side or covert channels. 
%
   Felten et al.~\cite{Felten:2000:TAW:352600.352606} first point out that due to the existence of web content caching, the loading time of external resource can be used to infer the browsing history in the past.  Then Bortz et al.~\cite{Bortz:2007:EPI:1242572.1242656} classify timing channels into two categories: direct timing attacks that infer private information stored at server side, and cross-site timing attacks that infer the size of a cross-site resource in the client browser.  The former is beyond the scope of the paper because the target secret  is caused by the server as discuss in Section~\ref{sec:threatmodel}; the latter is within the scope.  After that, Kotcher et al.~\cite{Kotcher:2013:CPS:2508859.2516712} find another timing attack, showing that the time of rendering a document after applying a CSS filter is related to the document's visual content.   Similarly, Goethem et al.~\cite{VanGoethem:2015:CST:2810103.2813632,vanrequest} show that the parsing time of scripts and video can be used to infer the size of the corresponding resource.  Oren et al.~\cite{Oren:2015:SSP:2810103.2813708} and Gras et al.~\cite{ndss17aslr} show that lower-level caching attacks can be launched from the JavaScript without privileged access. 





\subsection{Countermeasures of Timing Attacks}

We now discuss existing countermeasures of timing attacks. 

\subsubsection{Browser-level Defense} 
 Tor Browser~\cite{torbrowser} and Fermanta~\cite{fuzzyfox}  are the closest work to our deterministic browser.  Tor Browser~\cite{torbrowser}, an industry pioneer in fighting browser fingerprinting, reduces the clock resolution to 100ms, and adds noises at all places to mitigate the browser's fingerprintability.  Similarly, Fermata~\cite{fuzzyfox} reduces the clock resolution and adds several pause tasks in the event queue to reduce the resolution of implicit clocks. 
 The main difference with \sys 
   is that existing works only limit the attacks' capability, e.g., prolonging the attack or reducing the channel bandwidth, but do not fundamentally limit the timing attack.  
 
 Another difference between Fermata and deterministic browser is that Fermata requires that all synchronous JavaScript calls are converted to asynchronous because Fermata adds jitters in the event queue.  Such drastic changes will not be backward compatible with legacy JavaScript program and their implementation, i.e., FuzzyFox, does not support such changes. 
 

\subsubsection{System-level Defense (based on Determinism)} 
 Both StopWatch~\cite{Li:2014:SCA:2689660.2670940,DBLP:conf/dsn/LiGR13} and DeterLand~\cite{deterland} use determinism to prevent lower-level timing side or covert channels.  Their virtual time ticks based on the number of executed binary instructions in the virtual CPU.  Similarly, Burias et al.~\cite{buiras2013library} propose a deterministic information-flow control system to remove lower level cache attacks and then Stefan et al.~\cite{stefan2013eliminating} show that such cache attacks still exist given a reference clock. Aviram et al.~\cite{Aviram:2010:DTC:1866835.1866854} use provider-enforced deterministic execution to eliminate timing channels within a shared cloud domain.
 Compared with all existing determinisms, apart from the domain difference, i.e., browser v.s., lower level, the differences can be stated from three aspects.  First, the communication between the virtual and outside world in existing approaches is non-deterministic, i.e., they still add jitters or group event together to limit the bandwidth of side or covert channels.  Second, there is only one clock, i.e., the virtual time, defined in existing approaches, while our work introduces many clocks in different RFs.  

\subsubsection{Language-based Defense} 
 Many language-based defense~\cite{conf/csfw/HuismanWS06,DBLP:conf/csfw/ZdancewicM03,sabelfeld2000probabilistic,Smith:1998:SIF:268946.268975,volpano1997eliminating,DBLP:conf/pldi/ZhangAM12} have been proposed to provide well-typed language that either provably leaks a bounded amount of information or prevents timing attacks fundamentally.  Though effective, such approaches face backward compatibility problem, i.e., all the existing programs need to be rewritten and follow their specifications. 



\subsubsection{Detection of Timing Attacks} 
 Many approaches~\cite{Cabuk:2004:ICT:1030083.1030108,conf/ccs/GianvecchioW07,10.1109/SP.2006.28} focus on the detection of timing side or covert channels based on the extraction of high-level information, such as entropy.  Some existing work, namely Chen et al.~\cite{186208}, adopt determinism to replay and detect timing channels via non-deterministic events.  The difference between existing works in detecting timing attacks and our deterministic browser is that such existing work is reactive, i.e., waiting for timing attacks to happen and then detecting them, while our deterministic browser is proactive, i.e., preventing timing attacks from happening in the first place. 


\subsubsection{Defense against Web Tracking} 
 Some approaches, e.g., PriVaricator~\cite{Nikiforakis:2015:PDF:2736277.2741090} 
and TrackingFree~\cite{trackingfree}, aim to prevent web tracking in general.  The purpose of these approaches is different from ours. 
 JavaScript performance fingerprinting is just a small component of web tracking, and neither work can defend against JavaScript performance fingerprinting. 

\subsection{Other Similar Techniques}

We discuss some existing techniques that are similar to our deterministic browser.

\subsubsection{Determinism in General}  
 Determinism is also used in deterministic scheduling~\cite{parrot:sosp13,peregrine:sosp11,cui:tern:osdi10,smt:cacm,liu2011dthreads,olszewski2009kendo}, such deterministic multithreading.  
  Such techniques make the execution of programs, e.g., a  multithreaded one, follow a certain pattern so as to mitigate bugs, e.g., concurrency ones.  That is, they need to align execution sequence rather than adjusting the clocks.  As a comparison, our deterministic browser not only align execution sequence as in deterministic scheduling,  but also adjust different clocks in RFs. 
 
\subsubsection{Logical Time} 
 Logical time is also used in distributed systems to solve the problem of ``happening before'', such as 
Lamport Clock~\cite{Lamport:1978:TCO:359545.359563} and  
 Virtual Time~\cite{Jefferson:1985:VT:3916.3988}. The purpose of such work is to define a partial or total order of events so that causal dependency can be correctly resolved.  The reason for such logical time is that a distributed system lacks a centralized management; by contrast, a deterministic browser has a browser kernel to coordinate different RFs with stronger clock synchronization  than ordering events.


\section{Conclusion}

In conclusion, we propose deterministic browser, the first approach that introduces determinism into web browsers and provably prevents browser-related timing attacks.   
 Specifically, we break a web browser down into many small units, called reference frames (RFs).  In a RF, we can easily remove one of the three key elements, i.e., an adversary, a target secret and a reference clock, in timing attacks.  To achieve the removal purpose, we have two tasks: ($i$) making RFs with a target deterministic, and ($ii$) making the communication between RFs---especially a deterministic RF with a non-deterministic one---deterministic.  
 
 %
 We implemented a prototype of deterministic browser, called \sys, upon Firefox browser.  Our evaluation shows that \sys can defend against existing timing attacks in the literature, and is compatible with real-world websites.

\section{Acknowledgement}
The authors would like to thank anonymous reviewers for
their thoughtful comments. This work is supported in part
by U.S. National Science Foundation (NSF) under Grants
CNS-1646662 and CNS-1563843. The views and conclusions
contained herein are those of the authors and should not be
interpreted as necessarily representing the official policies or
endorsements, either expressed or implied, of NSF.

\bibliographystyle{ACM-Reference-Format}
\balance
\bibliography{reference,cao,websec,unlearning,r1}
\normalsize

\appendix
\section{Proof of Lemma 1} \label{lemma1}

\begin{proof}
 Let us look at the two conditions in Definition~\ref{def:comm} separately.  First, if the clocks in both RFs are synchronized at the moment of the communication, the sender does not convey additional timing information to an observer in the receiver.  Therefore, Definition~\ref{def:det} is satisfied.



Second, if the sender is deterministic, the clock information can be directly conveyed to the receiver.  Specifically, we will show that Definition~\ref{def:det} is satisfied in either clock or a mixture of both, which makes the receiver's RF still deterministic.  Because the receiver is deterministic in its own clock, we just need to consider the other two cases. 

($i$) Since the sender is deterministic, following Definition~\ref{def:det}, we have Equation~\ref{eq:sender}.
\begin{equation}
t^{sender}_2 - t^{sender}_1 = f(O^{sender}_1,O^{sender}_2) \label{eq:sender}
\end{equation}

Because the sender can convey any information in the message, both $O^{receiver}_1$ and $O^{receiver}_2$ contain information about $O^{sender}_1$ and $O^{sender}_2$.  That is, we will have Equation~\ref{eq:senderinrev}, meaning that the receiver is deterministic in the sender's clock. 
\begin{equation}
t^{sender}_2 - t^{sender}_1 = f^\prime(O^{receiver}_1,O^{receiver}_2) \label{eq:senderinrev}
\end{equation}

($ii$) We further show that the mixture of the sender's and the receiver's clock, i.e., $t^{sender}_2  - t^{receiver}_1$, is deterministic. According to Equation~\ref{eq:senderinrev} and the definition of the receiver's determinism, we have Equation~\ref{eq:mix1}.  

\begin{equation}
\begin{split}
&t^{sender}_2 - t^{sender}_{start} = f_1(O^{sender}_2, O^{sender}_{start})\\
&t^{receiver}_1 - t^{receiver}_{start} = f_2(O^{receiver}_1,O^{receiver}_{start}) \\
\end{split}\label{eq:mix1}
\end{equation}

When we minus the second equation in Equation~\ref{eq:mix1} from the first, we will have Equation~\ref{eq:mix2}.
\begin{equation}
\begin{split}
&(t^{sender}_2 - t^{sender}_{start}) - (t^{receiver}_1 - t^{receiver}_{start})\\
&= f_1(O^{sender}_2, O^{sender}_{start}) \\
& - f_2(O^{receiver}_1,O^{receiver}_{start})\\
\end{split}\label{eq:mix2}
\end{equation}

After doing some transformations in Equation~\ref{eq:mix2}, we will have Equation~\ref{eq:mix3}, following Definition~\ref{def:det} and showing that the receiver is deterministic in the clock mixing both the sender's and the receivers.  Note that both $O^{sender}_{start}$ and $O^{receiver}_{start}$ are constant. 
\begin{equation}
\begin{split}
&t^{sender}_2  - t^{receiver}_1 \\
&=t^{sender}_{start} - t^{receiver}_{start} +  f_1(O^{sender}_2, O^{sender}_{start}) \\
& - f_2(O^{receiver}_1,O^{receiver}_{start})\\
& = f^{\prime\prime}(O^{receiver}_{1},O^{sender}_2)
\end{split}
\label{eq:mix3}
\end{equation}


In sum, if a communication from a sender to a receiver obeys Definition~\ref{def:comm}, the communication does not break the determinism in the receiver. 

\end{proof}

\section{Proof of Theorem 1} \label{theorem1}


\begin{proof}
According to Lemma~\ref{theorem:detdefeattiming}, even if an adversary observer in a RF obtains an external clock, the RF is still deterministic.  Therefore, without loss of generality, we only consider one clock in the proof.  There are two sub-scenarios: (i) an observer in a RF measuring an internal target secret, and (ii) an observer in a RF measuring an external target secret in another RF.

($i$) Let us look at an observer in a RF measuring an internal target secret.  Let us assume that a target secret happens in two copies ($RF1$ and $RF2$) of the same RF, e.g., the execution runtime of the same JavaScript code with the same inputs. An observer ($Ob$) makes observations ($O^{RF1}_{\{St, End\}}$ and $O^{RF2}_{\{St, End\}}$) of the target secret in these two RFs.  We have the following Equation~\ref{eq:obeq}.
\begin{equation} \label{eq:obeq}
O^{RF1}_{End} = O^{RF2}_{End} \hspace{0.3in} O^{RF1}_{St} = O^{RF2}_{St} 
\end{equation}

With Equation~\ref{eq:obeq}, we have Equation~\ref{eq:determinismdefeattiming}. 

\begin{equation}\label{eq:determinismdefeattiming}
\begin{split}
&\Delta t^{RF1} - \Delta t^{RF2} = (t^{RF1}_{End} - t^{RF1}_{St}) - (t^{RF2}_{End} - t^{RF2}_{St}) \\
&= f(O^{RF1}_{End}, O^{RF1}_{St}) - f(O^{RF2}_{End}, O^{RF2}_{St})  \\
&= 0 \\
&\Delta t^{RF1} = \Delta t^{RF2} = const
\end{split}
\end{equation}

Because $RF1$ and $RF2$ are arbitrary RFs, we show that $\Delta t = t_{End} - t_{St} = const$

($ii$)  Let us then look at an observer in a RF measuring an external target secret in another RF.  Because according to Definition~\ref{def:comm}, the clocks in both RFs are synchronized,  the measurement that the observer makes is the same as another observer, maybe virtual, in the RF with a target secret.  Therefore, the sub-scenario just boils down the first one.

%


\end{proof}

\section{Proof of Theorem 2} \label{app:deteq}


\begin{proof}
We prove this theorem from two aspects:

\begin{icompact}
\item {Definition~\ref{def:det}$\Rightarrow$Definition~\ref{def:jsdet}:} 
 Say for two different executions ($E1$ and $E2$), i.e., two RFs in the context of JavaScript, an observer makes two observations at the starting point ($st$) of the execution and the timestamp of a specific opcode ($op$).  According to Definition~\ref{def:det}, we will have Equation~\ref{eq:opst1} and~\ref{eq:opst2}. 
\begin{equation}
t_{op}^1-t_{st}^1 = f(O_{op}^1, O_{st}^1)
\label{eq:opst1}
\end{equation}
\begin{equation} 
t_{op}^2-t_{st}^2 = f(O_{op}^2, O_{st}^2)
\label{eq:opst2}
\end{equation}

Then, we perform this operation (Equation~\ref{eq:opst1}$-$Equation~\ref{eq:opst2}):
\begin{equation}
\begin{split}
&t_{op}^1-t_{st}^1 - (t_{op}^2-t_{st}^2) \\
& = f(O_{op}^1, O_{st}^1) -  f(O_{op}^2, O_{st}^2) \\
\Rightarrow \enskip &t_{op}^1 - t_{op}^2 \\
&= t_{st}^1 - t_{st}^2 + f(O_{op}^1, O_{st}^1) -  f(O_{op}^2, O_{st}^2) \\
&= C\\
\end{split}\label{eq:opst3}
\end{equation}

Equation~\ref{eq:opst3} obeys Definition~\ref{def:jsdet}. \\

\item {Definition~\ref{def:jsdet}$\Rightarrow$Definition~\ref{def:det}:}  Say there is a specific execution ($E_{sd}$), i.e., a RF in JavaScript context, which can be used as a standard for comparison with any other executions.  The execution starts from $op_1$ to $op_n$. For an arbitrary opcode ($op_i$ where $i\in \{1...n\}$) in that execution, we have the following Equation~\ref{eq:sdexe}.
\begin{equation}
t_{op_i}^{E_{sd}} = \sum_{k\in \{1...i\}} t_{op_k}^{E_{sd}} \label{eq:sdexe}
\end{equation}

For an arbitrary execution ($E_{ab}$), a replicate of $E_{sd}$ but at a different time, according to Definition~\ref{def:jsdet}, we have Equation~\ref{eq:abstcomp}.
\begin{equation}
\begin{split}
&t_{op_i}^{E_{ab}} - t_{op_i}^{E_{sd}} = C \\
\Rightarrow \enskip &t_{op_i}^{E_{ab}} = t_{op_i}^{E_{sd}} + C = \sum_{k\in \{1...i\}} t_{op_k}^{E_{sd}} + C \\
\end{split} \label{eq:abstcomp}
\end{equation}

Now, say an observer makes two observations in $E_{ab}$, which corresponds to two opcodes ($op_i$ and $op_j$). We have Equation~\ref{eq:detproof}.
\begin{equation}
\begin{split}
t_{op_i}^{E_{ab}} - t_{op_j}^{E_{ab}} &= \sum_{k\in \{1...i\}} t_{op_k}^{E_{sd}} + C - (\sum_{k\in \{1...j\}} t_{op_k}^{E_{sd}} + C)  \\
 &= \sum_{k\in \{i...j\}} t_{op_k}^{E_{sd}} \\
 &= f(op_i, op_j) \\ 
 \end{split} \label{eq:detproof}
\end{equation}

Equation~\ref{eq:detproof} obeys Definition~\ref{def:det}. Note that the last step in  Equation~\ref{eq:detproof} is due to that $E_{ab}$ is a replication of $E_{sd}$ and the observations in a JavaScript RF are the executed opcodes.  

\end{icompact}


\end{proof}


\end{document}